\documentclass[12pt]{article}
\usepackage{amssymb}
\usepackage{amsthm}
\usepackage{amsmath}
\usepackage{algorithm} 
\usepackage{algorithmic} 
\usepackage{graphicx}
\usepackage{subfigure}
\newtheorem{lemma}{Lemma}[section]

\newtheorem{Theorem}[lemma]{Theorem}
\newtheorem{remark}[lemma]{Remark}

\begin{document}\baselineskip 0.56cm

\title{ Performance Guaranteed Evolutionary Algorithm for Minimum Connected Dominating Set}

\author{Chaojie Zhu$^1$\quad Yingli Ran$^1$ \quad Zhao Zhang$^1$\thanks{Corresponding author: Zhao Zhang, hxhzz@sina.com}\quad Ding-Zhu Du$^2$\\
\small $^1$ College of Mathematics and Computer Sciences, Zhejiang Normal University\\
\small Jinhua, Zhejiang, 321004, China\\
\small $^2$ Department of Computer Science, University of Texas at Dallas\\
\small Richardson, TX, 75080, USA}
\date{}

\maketitle

\begin{abstract}
A connected dominating set is a widely adopted model for the virtual backbone of a wireless sensor network. In this paper, we design an evolutionary algorithm for the minimum connected dominating set problem (MinCDS), whose performance is theoretically guaranteed in terms of both computation time and approximation ratio. Given a connected graph $G=(V,E)$, a connected dominating set (CDS) is a subset $C\subseteq V$ such that every vertex in $V\setminus C$ has a neighbor in $C$, and the subgraph of $G$ induced by $C$ is connected. The goal of MinCDS is to find a CDS of $G$ with the minimum cardinality. We show that our evolutionary algorithm can find a CDS in expected $O(n^3)$ time which approximates the optimal value within factor $(2+\ln\Delta)$, where $n$ and $\Delta$ are the number of vertices and the maximum degree of graph $G$, respectively.

\end{abstract}

\section{Introduction}\label{sec1}

Evolutionary algorithms have been widely used in real-world applications due to their simplicity of design and generality of implementation \cite{Neumann2,Yao}.
Although they have been demonstrated to be quite effective and efficient by empirical studies, their theoretical analysis is far behind. From the beginning of the 21st century, theoretical studies on evolutionary algorithms have begun to come onto stage \cite{Wegener}.
We are particularly interested in the application of evolutionary algorithms on NP-hard problems whose performance can be theoretically guaranteed. For example, in \cite{Friedrich}, Friedrich {\it et al.} designed a multi-objective evolutionary algorithm which can find an $O(\ln n)$-approximation for the minimum set cover problem (MinSC) in expected polynomial time. This ratio matches the lower bound for the approximability of MinSC. In \cite{Yu}, Yu {\it et al.} gave a general framework of an evolutionary algorithm, and when applied to the $k$-MinSC problem in which every set has size at most $k$, their algorithm can achieve in expected polynomial time the best known approximation ratio obtained in \cite{Levin}. However, we found that a lot of problems do not fall into these frameworks, including the minimum connected dominating set problem (MinCDS) studied in this paper. The motivation of this paper is to theoretically explore the power of evolutionary algorithm using the MinCDS problem as a wrench, trying to reveal more underlying relationship between evolutionary algorithm and approximation algorithm.

Given a graph $G=(V,E)$, a vertex set $C\subseteq V$ is a {\em dominating set} (DS) of $G$ if every vertex $v\in V\setminus C$ has a neighbor $u\in C$ ($v$ is said to be {\em dominated} by $u$, and $u$ is said to be a {\em dominator} of $v$). If furthermore, the subgraph of $G$ induced by $C$, denoted as $G[C]$, is connected, then $C$ is said to be a {\em connected dominating set} (CDS). The goal of the {\em minimum connected dominating set problem} (MinCDS) is to find a CDS with the minimum cardinality.

MinCDS has attracted extensive studies in the past two decades, due to its fundamental application in a wireless sensor network (WSN). From the beginning of this century, WSNs have become ubiquitous in various fields such as environmental monitoring, traffic control, production process, health detection, etc \cite{Wan}. Because sensors in a WSN have limited capacities, information is often forwarded through multi-hop transmissions. Notice that many sensors functioning at the same time may cause a large waste of energy, as well as a lot of interferences. So, researchers suggest using CDS as a {\em virtual backbone} in a WSN \cite{Das}: information collected by a non-backbone node can be forwarded to its dominator and then shared within the whole network since the backbone nodes form a connected structure.

The MinCDS problem is NP-hard \cite{Garey}. In fact, it is set-cover hard \cite{Guha} and thus does not admit a polynomial-time approximation within factor $(1-\varepsilon)\ln n$ for any real number $0<\varepsilon<1$ unless $P=NP$ \cite{Dinur}. On the other hand, $(\ln\Delta+2)$-approximation algorithm exists for MinCDS \cite{Ruan}, where $\Delta$ is the maximum degree of the graph. In view of its inapproximability, this ratio is asymptotically best possible under the assumption that $P\neq NP$.

In this paper, we show that the same approximation ratio $\ln\Delta+2$ can be achieved for MinCDS through a multi-objective evolutionary algorithm. The purpose of this paper is not to improve the approximation ratio since it is already the best possible. The aim is to understand the mechanism behind the evolutionary algorithm and its relation with the approximation algorithm.

\subsection{Related Works}

Evolutionary Algorithms (EAs) have their ideas originated in Darwin's theory of biological evolution. Basically speaking, an EA iteratively generates a set of offspring and retains only those individuals which are superior. There are a lot of versions of EAs, some of the simplest ones include single-objective versions like RLS (randomized local search) and $(1+1)$-EA \cite{Neumann} , multi-objective versions like SEMO and GSEMO \cite{Friedrich}.

The first origin of EAs can be dated to the early 1960s \cite{Tomassini},
and EAs have been applied to various areas such as antenna design \cite{Hornby}, bioinformatics \cite{Ma} and data mining \cite{Mukhopadhyay}. In fact, since EAs are general-purpose algorithms, which can be implemented without much knowledge about the concrete instances, they are favorable in engineering fields, especially when the structures of the problems are hard to obtain. Empirical studies show that they perform fairly well in real-world applications. On the contrary, theoretical analysis on EAs is far behind.

Recently, a lot of progress has been made on the running time and performance ratio analysis of EAs for a lot of combinatorial optimization problems. For example, Neumann {\it et al.} \cite{Neumann} showed that $(1+1)$-EA can find a minimum spanning tree (MST) in expected $O (m^2(\log n +\log w_{max}))$ iterations, where $m$ is the number of edges, $n$ is the number of vertices, and $w_{max}$ is the maximum edge weight of the graph. Friedrich {\it et al.} \cite{Friedrich} studied EAs for the minimum vertex cover problem (MinVC), which is a special case of the minimum set cover problem (MinSC). They  showed that even on a graph as simple as a complete bipartite graph, with a positive probability, RLS cannot find a solution within approximation ratio $\varepsilon/(1-\varepsilon)$ in finite time (and thus the approximation ratio can be arbitrarily bad), and it takes $(1+1)$-EA exponential time to obtain an optimal solution. On the contrary, GSEMO can find a $(\ln n+1)$-approximate solution for MinSC in expected $O (n^2m + mn(\log m + \log c_{max}))$ iterations, where $n$ is number of elements to be covered, $m$ is the number of sets, and $c_{max}$ is the maximum cost of a set. These studies demonstrate the advantage of multi-objective EAs over single-objective EAs, and lead to a series of follow-up studies using multi-objective EAs for coverage problems
\cite{Qian1,Qian2,Qian3}. In particular, Yu {\it et al.} \cite{Yu} introduced a framework of an evolutionary algorithm, and when applied to the minimum $k$-set cover problem (a MinSC problem in which every set has size at most $k$) can obtain ratio $H_k-\frac{k-1}{8k^9}$, where $H_k$ is the $k$th Harmonic number. However, there are a lot of combinatorial optimization problems (including the MinCDS problem studied in this paper) which do not fall into this framework.

There are also a lot of studies of evolutionary algorithms on various other combinatorial problems, including the maximum matching problem \cite{Giel}, the partition problem \cite{Witt}, the bi-objective pseudo-Boolean problem \cite{Giel0}, makespan scheduling \cite{Zhou}, shortest paths \cite{Doerr0}, etc. The reader may also refer to the monograph \cite{Neumann2}.

Although many works have been devoted to understanding the behavior of EAs from a theoretical point of view, there are large quantities of combinatorial problems which have not been sufficiently studied in view of approximations by EAs. It remains to be further explored which problems could be solved by EAs with a guaranteed theoretical performance. Finding underlying relationships between evolutionary algorithms and approximation algorithms is an ongoing challenge.

For the MinCDS problem, Guhn and Khuller \cite {Guha} provided a two-stage greedy algorithm with approximation ratio at most $H_{\Delta}+2$, where $\Delta$ is the maximum degree of the graph (notice that $\ln \Delta\leq H_{\Delta}\leq \ln \Delta+1$). This ratio was improved to $(2+\ln \Delta)$ by Ruan {\it et al.} \cite{Ruan} using a one-stage greedy algorithm. The most challenging part is how to deal with a non-submodular potential function. Non-submodularity is also the barrier which prevents us to use existing studies on EAs for submodular optimization problems on the MinCDS problem. New insights have to be explored.

\subsection{Our Contributions}

In this paper, we design an evolutionary algorithm for the MinCDS problem, and show that it can find a CDS in expected $O(n^3)$ iterations which approximates the optimal value within factor $\ln\Delta+2$. This ratio coincides with the best known ratio for MinCDS which is obtained by a centralized greedy algorithm. The main challenge lies in the theoretical analysis: can an evolutionary algorithm successfully simulate an approximation algorithm such that both efficiency and effectiveness can be guaranteed?

As shown in \cite{Friedrich}, a single-objective evolutionary algorithm seems hopeless in approximating NP-hard problems even on very simple examples. So, in this paper, we adopt a multi-objective method, trying to minimize two evaluation functions simultaneously. The first evaluation function $f_1$ is used to measure how far the individual is from a feasible solution. A vertex set is a CDS if and only if its corresponding individual has $f_1$-value being 2. The second evaluation function $f_2$ is the size of the vertex set corresponding to the individual, which is the objective of the MinCDS problem to be minimized. The algorithm follows the general framework of a multi-objective evolutionary algorithm: generating offspring randomly and keeping the Pareto front until some criterion is met.

Although the algorithm is simple, what is most challenging is how to analyze its performance theoretically. In the centralized algorithm designed by Ruan {\it et al.} in \cite{Ruan}, there exists a ``greedy path'' (by adding vertices one by one greedily), which leads to a CDS with approximation ratio at most $\ln\Delta+2$. If our algorithm can follow such a path in generating offspring, then we are done. The question is: does such a greedy path exist in the population generated by the algorithm? We think this might be the same question asked by Friedrich {\it et al.} in their work \cite{Friedrich}, in which it was shown that for the minimum set cover problem, allowing enough time, the population contains a path which is no worse than the greedy path generated by an approximation algorithm. Furthermore, Friedrich {\it et al.} designed a potential to ``track'' such a path. Their analysis heavily depends on the fact that the covering function is a submodular function. However, the MinCDS problem does not belong to the realm of submodular optimization, which makes the situation much more complicated.

To realize the above idea and overcome the above difficulties, we divide the analysis into two phases, and define two potentials to track a desirable path in each of these two phases. A challenge is how to coordinate these two potentials. For that purpose, a concept of {\em good offspring} is conceived and our analysis only tracks those {\em good events}.

The algorithm used in this paper is essentially SEMO or GSEMO. But the analysis is much more challenging. This is what we think most interesting: revealing deep theoretical mechanisms embedded in a simple algorithm.

The remaining part of this paper is organized as follows. Section \ref{sec2} introduces some terminologies and concepts used in this paper. The algorithm is presented in Section \ref{sec3}. A theoretical analysis is given in Section \ref{sec5}. Experiments are done in Section \ref{sec7}, which are compared with greedy algorithm. Section \ref{sec6} concludes the paper with some discussions on future work.

\section{Preliminaries}\label{sec2}


In this paper, small bold letters (such as ${\bf x,y,z}$) are used to represent vectors.
Suppose $G=(V,E)$ is a graph with vertex set $V=\{v_1,\ldots,v_n\}$ and edge set $E$. Notice that there is a one-to-one correspondence between $2^V$ and $\{0,1\}^n$: for a vertex subset $C\subseteq V$, the vector ${\bf x}_C\in\{0,1\}^n$ corresponding to $C$ has its bits indexed by $v_1,\ldots,v_n$ and the $i$-th bit of ${\bf x}_C$ is $1$ if and only if $v_i\in C$; for a vector ${\bf x}\in\{0,1\}^n$, the vertex subset $C_{\bf x}$ corresponding to $\bf x$ consists of all those vertices with bits 1 in $\bf x$. Since this paper considers evolutionary algorithm, we call such a vector as an {\em individual}.

In a multi-objective optimization problem, different objectives might conflict with each other. To compare two individuals in such a setting, dominance is a commonly used terminology. In a multi-objective optimization problem in which $t$ objectives $f_1({\bf x}),\ldots,f_t({\bf x})$ are to be minimized, an individual ${\bf x}$ is said to  {\em weakly dominate} an individual ${\bf x}'$,  or we say that ${\bf x}$ is {\em weakly better} than ${\bf x}'$, denoted as ${\bf x} \succeq {\bf x'}$, if $f_i({\bf x})\leq f_i({\bf x'})$ holds for each $i=1,\ldots,t$. If furthermore, there is some $i$ with $f_i({\bf x})<f_i({\bf x'})$, then ${\bf x}$ is said to {\em dominate} ${\bf x'}$, or we say that $\bf x$ is {\em better} than ${\bf x}'$, denoted as ${\bf x}\succ {\bf x'}$. Notice that there exist individuals which are incomparable in the above sense: $\bf x$ might be superior to $\bf x'$ in terms of $f_i$ but inferior to $\bf x'$ in terms of $f_j$. Two individuals are said to be {\em incomparable} if no one is weakly better than the other. A {\em Pareto front} is a set of incomparable individuals.

Since the topic of this paper is dominating set, using terminology ``dominate'' in both the graph theoretic setting and the multi-objective optimization setting may incur confusion. So this paper only uses ``dominate'' in the graph theoretic setting, and uses ``better'' and ``weakly better'' in the multi-objective optimization setting.

There are a lot of methods to mutate an individual in evolutionary algorithms. We only consider two basic mutation methods: for a $0$-$1$ vector of dimension $n$, the first one flips exactly one bit uniformly at random; while in the second one, every bit is flipped independently with probability $1/n$.

In a simple single-objective evolutionary algorithm with objective function $f$, every iteration creates one offspring ${\bf x}'$ of current individual $\bf x$ and maintains the one in $\{{\bf x},{\bf x}'\}$ having the better $f$-value. If the first mutation is used, then it is known as {\em random local search} (RLS). If the second mutation is used, then it is known as {\em $(1+1)$-EA}.

Different from RLS and $(1+1)$-EA, a simple multi-objective evolutionary algorithm maintains the Pareto front $P$ of those individuals which have been created up to now, and in each iteration, an individual in $P$ is picked uniformly at random as the parent to create one offspring. In \cite{Neumann0}, such multi-objective evolutionary algorithms using the first mutation and the second mutation are called SEMO ({\em simple evolutionary multi-objective optimizer}) and GSEMO ({\em global SEMO}), respectively. As shown in \cite{Friedrich}, transforming a single-objective problem into a multi-objective problem and then using SEMO or GSEMO on it can significantly improve the performance of the solution.

The quality of an algorithm for an NP-hard problem is usually measured by its {\em efficiency} (that is, how long it takes to output a solution) and {\em effectiveness} (that is, how accurate the solution could be). We measure the running time by the expected number of mutations taken in the algorithm, and measure the accuracy by the approximation ratio. For a minimization problem, an algorithm is said to be an {\em $\alpha$-approximation algorithm} if for any instance of the problem, it can output a solution whose cost is at most $\alpha$ times that of an optimal solution.


\section{Algorithm for MinCDS}\label{sec3}

Given a graph $G=(V,E)$ on $n$ vertices, a vertex set $C$ can be represented by a vector ${\bf x}_C\in\{0,1\}^n$, and a vector ${\bf x}\in\{0,1\}^n$ corresponds to a vertex set $C_{\bf x}$. Use $|{\bf x}|$ to denote the number of 1's in ${\bf x}$, which is also the number of vertices in vertex set $C_{\bf x}$.

The algorithm is essentially SEMO, with two objective functions $f_1$ and $f_2$ to be minimized simultaneously. For a vector ${\bf x}\in\{0,1\}^n$, $f_2({\bf x})=|{\bf x}|$. So minimizing $f_2$ is to minimize the size of the solution. Next, we introduce $f_1$ which is used to measure the feasibility of the solution.

For a vertex subset $C\subseteq V$, let $G[C]$ be the subgraph of $G$ induced by $C$, and let $p(C)$ be the number of connected components of $G[C]$. Denote by $E_C$ the set of edges incident with $C$, that is, every edge in $E_C$ has at least one end in $C$. The spanning subgraph of $G$ on vertex set $V$ and edge set $E_C$ is denoted as $G\langle C\rangle$. Let $q(C)$ be the number of connected components of $G\langle C\rangle$. Consider a path on five vertices as an example (see Fig. \ref{fig0721-1} $(a)$). For vertex subset $C=\{v_1,v_5\}$, graph $G[C]$ consists of two singletons, namely $v_1$ and $v_5$, and thus $p(C)=2$. While $G\langle C\rangle$ is depicted in Fig. \ref{fig0721-1} $(b)$, and thus $q(C)=3$ (notice that $G\langle C\rangle$ is a {\em spanning} subgraph of $G$).

\begin{figure}[htpb]
\begin{center}
\begin{picture}(100,30)
\multiput(0,15)(25,0){5}{\circle*{5}}
\put(0,15){\line(1,0){100}}
\put(-4,20){$v_1$}\put(21,20){$v_2$}\put(46,20){$v_3$}\put(71,20){$v_4$}\put(96,20){$v_5$}
\put(42,0){(a)}
\end{picture}
\hskip 2cm\begin{picture}(100,30)
\multiput(0,15)(25,0){5}{\circle*{5}}
\put(0,15){\line(1,0){25}}\put(75,15){\line(1,0){25}}
\put(-4,20){$v_1$}\put(21,20){$v_2$}\put(46,20){$v_3$}\put(71,20){$v_4$}\put(96,20){$v_5$}
\put(42,0){(b)}
\end{picture}
\caption{An illustration for the definition of $p(C)$ and $q(C)$.}\label{fig0721-1}
\end{center}
\end{figure}

In the following, we shall use a vector and its corresponding vertex set interchangeably. For a vector ${\bf x}\in\{0,1\}^n$, let $f_1({\bf x})=p({\bf x})+q({\bf x})$. The next lemma shows how $f_1$ can be used to measure the feasibility of the solution.

\begin{lemma}\label{lem0713-6}
A vertex set $C$ is a CDS of $G$ if and only if $f_1({\bf x}_C)=2$.
\end{lemma}
\begin{proof}
Except for $\emptyset$, every vertex set $C$ has $p(C)\geq 1$ and $q(C)\geq 1$. So, $f_1(C)\geq 2$ and equality holds if and only if $p(C)=q(C)=1$. A CDS $C$ clearly satisfies $p(C)=q(C)=1$ and thus $f_1(C)=2$. On the other hand, if $C$ is not a dominating set, then there is a vertex $v\in V\setminus C$ which has no neighbor in $C$. In this case, $v$ forms a connected component itself in $G\langle C\rangle$, and thus $G\langle C\rangle$ has at least two connected components. So, $q(C)=1$ holds only when $C$ is a dominating set. Furthermore, the condition $p(C)=1$ is equivalent to saying that $G[C]$ is connected. So a vertex set $C$ with $p(C)=q(C)=1$ must be a CDS.
\end{proof}

Let ${\bf x}^{(0)}=(0,0,\ldots,0)$ be the vector corresponding to vertex set $\emptyset$. Since $G\langle\emptyset\rangle$ consists of $n$ singletons, we have $q(\emptyset)=n$. Clearly, $p(\emptyset)=0$. So $f_1({\bf x}^{(0)})=n$ while $f_2({\bf x}^{(0)})=0$. We study the MinCDS problem by considering the bi-objective minimization problem $\min_{{\bf x}\in \{0,1\}^n}\{(f_1({\bf x}),f_2({\bf x}))\}$.

The algorithm is described in Algorithm \ref{alg2}. It starts from ${\bf x}^{(0)}$, and maintains a {\em population set} $P$ which keeps all those incomparable individuals found so far (namely Pareto front). In each iteration, the algorithm picks an individual $\bf x$ in $P$ uniformly at random, and generates an offspring $\bf x'$ by flipping one bit of $\bf x$ uniformly at random. If there is no individual in $P$ which is better than $\bf x'$, then $\bf x'$ enters $P$, and all those individuals which are weakly worse than $\bf x'$ are removed from $P$. We shall show that after some polynomial number of iterations, the population $P$ contains a desired CDS. The output is a vertex set $C$ corresponding to an individual ${\bf x}_C$ in the final population $P$ with the minimum $f_2$-value whose $f_1$-value equals $2$ (so $C$ is a CDS).

\begin{algorithm}[H]
\caption{Algorithm for MinCDS}
\textbf{Input:} A connected graph $G$, the number of iterations $T$.

\textbf{Output:} A set $C$ which is a CDS of $G$.
	
\begin{algorithmic}[1]\label{alg2}
\STATE $P \leftarrow {\bf x^0}$.
\STATE $t=0$
\WHILE{$t\leq T$}
    \STATE choose ${\bf x}$ from $P$ uniformly at random.\label{line3}
    \STATE create ${\bf x'}$ by flipping one bit of $\bf x$ uniformly at random.\label{line4}
    \IF{$ \nexists  {\bf z} \in P$ with ${\bf z} \succ {\bf x'}$}\label{line6}
    \STATE $P \leftarrow P\cup\{{\bf x'}\} \setminus \{{\bf z}: {\bf x'} \succeq {\bf z}, {\bf z}\in P\}$\label{line7}
    \ENDIF
    \STATE $t=t+1$
\ENDWHILE
\STATE ${\bf x}_C=\arg\{f_2({\bf x}): {\bf x}\in P,f_1({\bf x})=2\}$
\RETURN Vertex set $C$ corresponding to ${\bf x}_C$.
\end{algorithmic}
\end{algorithm}

\section{The Analysis}\label{sec5}

The following lemma is crucial to the analysis of our algorithm, which shows that there is a vertex the addition of which may decrease the $f_1$-value by roughly a multiplicative factor. The proof is similar to the one in \cite{Ruan}, we provide it here for the completeness of this paper.

\begin{lemma} \label{lem1}
For any vertex set $C$ in a connected graph $G=(V, E)$ with $f_1(C)>2$, let $v_C$ be the vertex of $V\setminus C$ satisfying $v_C=\arg\max_{v\in V\setminus C}\{f_1(C)-f_1(C\cup \{v\})\}$, then we have the following two properties:

$(\romannumeral1)$ $f_1(C\cup\{v_C\})\leq f_1(C)-1$, and

$(\romannumeral2)$ $\displaystyle f_1(C\cup\{v_C\})\leq \left(1-\frac{1}{m}\right)f_1(C)+\frac{2}{m}+1$,\\
where $m$ is the size of a minimum CDS.
\end{lemma}

\begin{proof}
$(\romannumeral1)$ We prove property $(\romannumeral1)$ by distinguishing two cases.

In the case $q(C)>1$, $C$ is not a dominating set by the proof of Lemma \ref{lem0713-6}. By the connectedness of $G$, there exists a vertex $v\in V\setminus C$ such that $v$ is not adjacent with $C$ and $v$ is adjacent with a vertex $u\in N(C)$ where $N(C)$ is the neighbor set of $C$. Then $p(C\cup\{u\})\leq p(C)$ and $q(C\cup\{u\})<q(C)$.

In the case $q(C)=1$, since $f_1(C)>2$, we have $p(C)>1$. So, $G[C]$ has at least two connected components. By $q(C)=1$, it can be seen that the closest connected components of $G[C]$ is exactly 2-hops away from each other, which implies that there is a vertex $v\in V\setminus C$ that is adjacent to at least two connected components of $G[C]$. In this case, $p(C\cup\{v\})< p(C)$ and $q(C\cup\{v\})=q(C)=1$.

In any case, $f_1(C\cup\{v\})< f_1(C)$. By the choice of $v_C$, and because $f_1$ is an integer-valued set function, we have $f_1(C\cup\{v_C\})\leq f_1(C\cup\{v\})\leq f_1(C)-1$.

$(\romannumeral2)$ In order to prove $(\romannumeral2)$, we show that there is a vertex in an optimal solution the addition of which satisfies the claimed inequality.

Let $C^*$ be a minimum CDS of $G$. Since $G[C^*]$ is connected, we may order vertices in $C^*$ as $v_1^*, v_2^*,\ldots, v_m^*$ such that for any $j=1,\ldots,m$, the subgraph of $G$ induced by vertex set $C_j^*=\{v_1^*, v_2^*,\ldots, v_j^*\}$ is connected. This can be done by finding a spanning tree of $G[C^*]$ and removing leaf vertices one by one, ordering the vertices in the reverse order of their removal. Due to the connectedness of $G[C_{j-1}^*]$, it can be seen that
$$
p(C\cup C_{j-1}^*)-p(C\cup C_j^*)\leq \big(p(C)-p(C\cup \{v_j^*\})\big)+1.
$$
Furthermore, it could be proved that function $-q$ is submodular and thus
$$
q(C\cup C_{j-1}^*)-q(C\cup C_j^*)\leq q(C)-q(C\cup \{v_j^*\}).
$$
It follows that
\begin{equation}\label{eq4}
f_1(C\cup C_{j-1}^*)-f_1(C\cup C_j^*)\leq \big(f_1(C)-f_1(C\cup \{v_j^*\})\big)+1.
\end{equation}

Rewrite $f_1(C)-f_1(C\cup C^*)$ in the following way:
$$
f_1(C)-f_1(C\cup C^*)=\sum \limits_{j=1}^m \big(f_1(C\cup C_{j-1}^*)-f_1(C\cup C_j^*)\big),
$$
where $C_0^*=\emptyset$. Making use of inequality \eqref{eq4}, and the fact $f(C\cup C^*)=2$ (because $C\cup C^*$ is a CDS), we have
$$
f_1(C)-2\leq \sum \limits_{j=1}^m \big(f_1(C)-f_1(C\cup \{v_j^*\})+1\big).
$$
Let $v_{j_0}^*=\arg\max_{v_j^*}\{f_1(C)-f_1(C\cup \{v_j^*\})\}$. Then
$$
f_1(C)-f_1(C\cup \{v_{j_0}^*\})+1\geq\frac{f_1(C)-2}{m}.
$$

By the choice of $v_C$, we have
$$
f_1(C)-f_1(C\cup \{v_C\})+1\geq f_1(C)-f_1(C\cup \{v_{j_0}^*\})+1\geq\frac{f_1(C)-2}{m},
$$
and thus $f_1(C\cup\{v_C\})\leq \left(1-\frac{1}{m}\right)f_1(C)+\frac{2}{m}+1$. Property $(\romannumeral2)$ is proved.
\end{proof}

The next theorem analyzes the expected time to obtain a CDS with a guaranteed approximation ratio.

\begin{Theorem}\label{thm2}
In expected $n(n-1)(n-2)$ iterations, the output $C$ of Algorithm \ref{alg2} is a CDS whose size is at most $(2+\ln\Delta)$ times of the optimal value, where $\Delta$ is the maximum degree of graph $G$, $n$ is the number of vertices in $G$.
\end{Theorem}
\begin{proof}
Let $m$ be the size of a minimum CDS of $G$. Notice that although we use the optimal value $m$ in the analysis, the algorithm does not depend on $m$.

For $i=2,3,\ldots,n$, let $B_i$ be the {\em bin} which contains the individual in $P$ whose $f_1$-value equals $i$. Notice that at any time, each bin contains at most one individual, this is because for any two individuals created in the evolutionary process, if they have the same $f_1$-value, then only the one with the smaller $f_2$-value is kept in $P$. To track a ``path'' of desirable mutations, we only consider those bins containing good individuals, where an individual $\bf x$ is said to be ``good'' if it satisfies the following constraint:
\begin{equation}\label{eq0701-1}
f_1({\bf x})\leq (n-2-m)\left(1-\frac{1}{m}\right)^{|{\bf x}|}+m+2.
\end{equation}

{\em Claim 1.} Any good individual ${\bf x}$ with $f_1({\bf x})\geq 2m+2$ has $|{\bf x}|\leq m\ln\Delta$.

By the definition of good individual,
$$
2m+2\leq f_1({\bf x})\leq (n-2-m)\left(1-\frac{1}{m}\right)^{|{\bf x}|}+m+2.
$$
Using $(1-\frac{1}{m})^{|{\bf x}|}\leq e^{-\frac{|{\bf x}|}{m}}$, it can be derived that
$$
|{\bf x}|\leq m \ln \left(\frac{n-2-m}{m}\right).
$$
Since every vertex can dominate at most $\Delta+1$ vertices, we have $n\leq (\Delta +1)m$, and thus $\frac{n-2-m}{m}\leq \Delta$. Hence $|{\bf x}|\leq m \ln \Delta$, Claim 1 is proved.

{\em Claim 2.} Suppose $\bf x$ is a good individual. Then, an individual ${\bf x}'$ with $f_1({\bf x}')\leq f_1({\bf x})$ and $|{\bf x}'|\leq |{\bf x}|$ is also good.

This is because
\begin{align}\label{eq0713-10}
f_1({\bf x}') & \leq f_1({\bf x})\leq (n-2-m)\left(1-\frac{1}{m}\right)^{|{\bf x}|}+m+2\nonumber\\
& \leq (n-2-m)\left(1-\frac{1}{m}\right)^{|{\bf x}'|}+m+2.
\end{align}

{\em Claim 3.} Suppose $\bf x$ is a good individual, let ${\bf x}'$ be the vector corresponding to vertex set $C_{\bf x}\cup \{v_{C_{\bf x}}\}$, where $v_{C_{{\bf x}}}=\arg\max_{v\in V\setminus C_{{\bf x}}}\{f_1(C_{{\bf x}})-f_1(C_{{\bf x}}\cup \{v\})\}$. Then

$(2.1)$ $|{\bf x}'|=|{\bf x}|+1$,

$(2.2)$ $f_1({\bf x'})\leq f_1({\bf x})-1$, and

$(2.3)$ ${\bf x}'$ is also a good individual.

Property $(2.1)$ is obvious, and property $(2.2)$ is a consequence of Lemma \ref{lem1} $(\romannumeral1)$. Next, we show property $(2.3)$. By Lemma \ref{lem1} $(\romannumeral2)$,
$$
f_1(C_{{\bf x}}\cup\{v_{C_{{\bf x}}}\})\leq \left(1-\frac{1}{m}\right)f_1({\bf x})+\frac{2}{m}+1.
$$
Since $\bf x$ is good and $|{\bf x}'|=|{\bf x}|+1$, we have
\begin{align*}
f_1({\bf x'}) & \leq \left(1-\frac{1}{m}\right)\left((n-2-m)\left(1-\frac{1}{m}\right)^{|{\bf x}|}+m+2\right)+\frac{2}{m}+1\\
& =(n-2-m)\left(1-\frac{1}{m}\right)^{|{\bf x'}|}+m+2.
\end{align*}
Hence $\bf x'$ is also a good individual.

An individual ${\bf x}'$ generated by the way described in Claim 2 is desirable, for convenience, we call such ${\bf x}'$ to be the {\em best offspring} of $\bf x$.

Next, we define an indicator $J_i$ for each bin $B_i$ such that individuals contained in bins with indicator 1 are all good.

Initially, only $J_n=1$ and all other $J_i=0$ for $i=2,3,\ldots,n-1$.

When we say that bin $B_i$ {\em is going to be filled with an individual ${\bf x}'$,} it refers to the setting that the newly generated individual ${\bf x}'$ can be put into $B_i$, either because $B_i$ is empty and $i=f_1({\bf x}')$, or because ${\bf x}'$ can replace the individual ${\bf x}$ in $B_i$ due to $f_1({\bf x}')=f_1({\bf x})=i$ and $|{\bf x}'|\leq |{\bf x}|$.

When bin $B_i$ is going to be filled with an individual ${\bf x}'$, its indicator $J_i$ is changed from 0 to 1 if one of the following two {\em good events} occur.

$(\romannumeral1)$ There is an individual $\bar{\bf x}$ belonging to a bin $B_j$ with indicator $J_j=1$ such that its best offspring $\bar{\bf x}'$ is weakly worse than ${\bf x}'$, that is
\begin{equation}\label{eq0713-3}
f_1(\bar{\bf x}')\geq f_1({\bf x}')\ \mbox{and}\ |\bar{\bf x}'|\geq |{\bf x}'|.
\end{equation}

$(\romannumeral2)$ There is an individual $\bar{\bf x}$ belonging to a bin $B_j$ with indicator $J_j=1$ such that
\begin{equation}\label{eq0713-9}
f_1(\bar{\bf x})>f_1({\bf x}')\ \mbox{and}\ |\bar{\bf x}|\geq |{\bf x}'|.
\end{equation}

It is possible that indicator $J_i$ can be changed from 1 to 0: if the individual contained in bin $B_i$ is removed from the population because of the creation of a better individual, then $J_i$ is set to be 0.

{\em Claim 4.} Any bin $B_i$ with $J_i=1$ contains a good individual.

We first show that if current $B_i$ contains a good individual, then as long as $B_i$ continues to be non-empty, the individual contained in $B_i$ is always good. In fact, suppose $\bf x$ is a good individual in $B_i$, it can be replaced by another individual ${\bf x}'$ only when $f_1({\bf x'})=f_1({\bf x})$ and $|{\bf x'}|\leq |{\bf x}|$. Then by Claim 2, $\bf x'$ is also good.

So, as long as we can show the following property:
\begin{equation}\label{eq0713-11}
\mbox{the time when $J_i$ changes from 0 to 1, the individual entering $B_i$ is good,}
\end{equation}
then all individuals contained in $B_i$ during the span when $J_i=1$ will be good. Next, we use induction on the process of assigning indicator 1's to prove \eqref{eq0713-11}.

Initially, only $J_n=1$, and $B_n=\{{\bf x}^{(0)}\}$. Notice that $|{\bf x^0}|=0$ and $f_1({\bf x^0})=n=(n-2-m)(1-\frac{1}{m})^0+m+2$, so ${\bf x}^{(0)}$ is a good individual.

Suppose the claim is true for all bins with indicator 1 in current population, and after some mutation, $J_i$ is to be changed from 0 to 1. The first case for such a change is because of the existence of ${\bf x}',\bar{\bf x},\bar{\bf x}'$ satisfying \eqref{eq0713-3}. By induction hypothesis, $\bar{\bf x}$ is good. By Claim 3, $\bar{\bf x}'$ is also good. Then by condition \eqref{eq0713-3} and Claim 2, ${\bf x}'$ is also good. The second case for such a change is because of the existence of ${\bf x}',\bar{\bf x}$ satisfying \eqref{eq0713-9}. By induction hypothesis, $\bar{\bf x}$ is good, and thus by Claim 2, ${\bf x}'$ is good. Claim 4 is proved.

The following analysis is divided into two phases. For each phase, a potential will be used to track the progress of the algorithm.

Define a potential $I$ with respect to current population $P$ to be the integer
\begin{equation}\label{eq0713-2}
I=\arg\min\{i\colon J_i=1\}.
\end{equation}
The initial population $P=\{{\bf x}^{(0)}\}$ has potential $I=n$.

{\em Claim 5.} Potential $I$ will never increase, and the expected time for $I$ to be decreased by at least 1 is at most $n(n-1)$.

Suppose $\bar{\bf x}$ is the individual corresponding to current potential $I$. As long as $\bar{\bf x}$ remains in $P$, the value of $I$ will not increase. By the algorithm, $\bar{\bf x}$ can be removed from $P$ only when a newly generated individual ${\bf x'}$ satisfies $f_1({\bf x'})\leq f_1(\bar{\bf x})$ and $|{\bf x'}|\leq |\bar{\bf x}|$. If $f_1({\bf x}')=f_1(\bar{\bf x})$, then $\bar{\bf x}$ is replaced by ${\bf x}'$ and $I$ remains the same. If $f_1({\bf x}')<f_1(\bar{\bf x})$, then by \eqref{eq0713-9}, the birth of ${\bf x}'$ makes $J_{i'}$ to change from 0 to 1, where $i'=f_1({\bf x}')<f_1(\bar{\bf x})=i$. In this case, the new $I'=i'<I$. In any case, $I$ will not increase.

Next, we consider the expected time to reduce $I$. If the mutation picks $\bar{\bf x}$ and generates its best offspring $\bar{\bf x}'$, then $J_i$ can be changed from 0 to 1, where $i=f_1(\bar{\bf x}')\leq f_1(\bar{\bf x})-1<I$ (notice that before the mutation, $J_i$ must be $0$ since $I$ is the {\em smallest} index for the bin with indicator 1). Notice that the probability of choosing $\bar{\bf x}$ in line \ref{line3} of Algorithm \ref{alg2} is $\frac{1}{|P|}$, and the probability of obtaining its best offspring $\bar{\bf x}'$ in line \ref{line4} of Algorithm \ref{alg2} is $\frac{1}{n}$. So, the probability
\begin{equation}\label{eq0703-1}
Pr[\mbox{$I$ is reduced}]\geq \frac{1}{n|P|}\geq \frac{1}{n(n-1)},
\end{equation}
where $|P|\leq n-1$ is used since $f_1$ has at most $n-1$ distinct values, namely $2,3,\ldots,n$, and no solutions having the same $f_1$-value can coexist in $P$ by line \ref{line7} of the algorithm.

Notice that the derivation of \eqref{eq0703-1} is valid starting from ``any'' individual in the population achieving current potential, even when the above $\bar{\bf x}$ is removed from $P$ due to the entering of a better individual. So, the expected number of mutations for potential $I$ to be decreased by at least $1$ is at most $n(n-1)$. Claim 5 is proved.

By Claim 5, in expected $n(n-1)(n-2)$ iterations, $I$ could be reduced from $n$ to $2$. However, merely tracking potential $I$ could not guarantee that the final solution has small size. In fact, since $f_1(V)=2$, the vector ${\bf x}^{(1)}=(1,1,\ldots,1)$ satisfies \eqref{eq0701-1}, and thus $V$ is also a good individual, whose size is clearly too large. From the algorithmic point of view, when $f_1({\bf x})$ is small, the decreasing speed of $f_1$ implied by Lemma \ref{lem1} $(\romannumeral2)$ will slow down. So, we shall track another potential after the $f_1$-value becomes smaller than $2m+2$. The details are given as follows.

If $n<2m+2$, then the approximation ratio $2+\ln\Delta$ holds for the trivial solution $V$. So, we assume $n\geq 2m+2$ in the following.  Since the initial value for $I$ is $n\geq 2m+2$, and the final value for $I$ is $2<2m+2$, there must be a time when $I$ jumps over $2m+2$. That is, there exists a mutation, before which potential $\widetilde I\geq 2m+2$, after which potential $\widetilde I'\leq 2m+1$.

Define a second potential $K$ as follows.
\begin{equation}\label{eq0708-1}
K=\arg \min \limits_{{\bf x}\in P} \{f_1({\bf x}): |{\bf x}|+f_1({\bf x})\leq m \ln \Delta +2m+2\}.
\end{equation}

{\em Claim 6.} Let $\widetilde{\bf x}'$ be the individual corresponding to $\widetilde I'$. Then $\widetilde{\bf x}'$ satisfies the constraint in \eqref{eq0708-1}. Hence from the time when $\widetilde{\bf x}'$ enters population $P$, potential $K$ is well-defined, and $K\leq f_1(\widetilde{\bf x}')=\widetilde I'\leq 2m+1$.

The core part for the proof of Claim 6 is to show that $\widetilde{\bf x}'$ satisfies the constraint in \eqref{eq0708-1}. Since we are considering the time when potential $\widetilde I\geq 2m+2$ is jumped down to $\widetilde I'=f_1(\widetilde{\bf x}')\leq 2m+1$, the creation of $\widetilde{\bf x}'$ incurs the change of $J_{\widetilde I'}$ from 0 to 1. The reason for such a change is because of the existence of a bin $B_j$ with $J_j=1$ and the individual $\bar{\bf x}$ in $B_j$ satisfying constraint \eqref{eq0713-3} or constraint \eqref{eq0713-9} (with ${\bf x}'$ replaced by $\widetilde{\bf x}'$), that is,
\begin{equation}\label{eq0713-8}
f_1(\bar{\bf x}')\geq f_1(\widetilde{\bf x}')=\widetilde I'\ \mbox{and}\ |\bar{\bf x}'|\geq |\widetilde {\bf x}'|
\end{equation}
or
\begin{equation}\label{eq0714-1}
f_1(\bar{\bf x})>f_1(\widetilde{\bf x}')\ \mbox{and}\ |\bar{\bf x}|\geq |\widetilde{\bf x}'|.
\end{equation}
Since $J_j=1$, by Claim 4, $\bar{\bf x}$ is good. Since $\widetilde I$ is the smallest index of bins with indicator 1 before the mutation, we have $j\geq \widetilde I\geq 2m+2$. So $\bar{\bf x}$ is a good individual with $f_1(\bar{\bf x})=j\geq 2m+2$. Then by Claim 1,
\begin{equation}\label{eq0714-2}
|\bar{\bf x}|\leq m\ln\Delta.
\end{equation}
Since $\bar{\bf x}'$ is the best offspring of $\bar{\bf x}$, we have
$$
|\bar{\bf x}'|=|\bar{\bf x}|+1\leq m\ln\Delta+1.
$$
So, in the first case, namely \eqref{eq0713-8}, we have
$$
|\widetilde{\bf x}'|+f_1(\widetilde{\bf x}')\leq |\bar{\bf x}'|+\widetilde I'\leq (m\ln\Delta+1)+(2m+1)=m\ln\Delta+2m+2.
$$
The second case \eqref{eq0714-1} is easier by using \eqref{eq0714-2}. Claim 6 is proved.

{\em Claim 7.} The potential $K$ never increases and the expected time for $K$ to be decreased by at least 1 is at most $n(n-1)$.

To prove the monotonicity of $K$, similar to the proof of Claim 5, it suffices to consider a solution $\bar{\bf x}$ corresponding to current $K$, and the case that $\bar{\bf x}$ is removed from $P$. In such a case, a newly generated solution ${\bf x'}$ satisfies $f_1({\bf x'})\leq f_1(\bar{\bf x})$ and $|{\bf x'}|\leq |\bar{\bf x}|$. Such ${\bf x'}$ satisfies $|{\bf x'}|+f_1({\bf x'})\leq|\bar{\bf x}|+f_1(\bar{\bf x})\leq m \ln \Delta +2m+2$. So $\bf x'$ also satisfies the constraint in \eqref{eq0708-1}, and thus replacing $\bar{\bf x}$ by ${\bf x'}$  will not increase $K$.

To prove the second part of Claim 7, observe that if a mutation picks $\bar{\bf x}$ which corresponds to current $K$ and generates its best offspring $\bar{\bf x}'$, then by Claim 3, $f_1(\bar{\bf x}')\leq f_1(\bar{\bf x})-1<K$ and $|\bar{\bf x}'|=|\bar{\bf x}|+1$. It follows that $|\bar{\bf x}'|+f_1(\bar{\bf x}')\leq |\bar{\bf x}|+f_1(\bar{\bf x})\leq m \ln \Delta +2m+2$ and thus $\bar{\bf x}'$ also satisfies the constraint in \eqref{eq0708-1}. Hence the new potential $K'\leq f_1(\bar{\bf x}')<K$. The expected time for the occurrence of such an event is at most $n(n-1)$, as in the proof of Claim 5. So, Claim 7 is proved.

{\em Claim 8.} Starting from initial individual ${\bf x}^{(0)}$, the expected time for $K$ reaching 2 is at most $n(n-1)(n-2)$.

By Claim 5, the expected time for potential $I$ to decrease from $f_1({\bf x}^{(0)})=n$ to $\widetilde I'=f_1(\widetilde{\bf x}')$ is at most $\big(n-\widetilde I'\big)n(n-1)$. Then by Claim 7, potential $K$ becomes well-defined and keeps decreasing, and the the expected time for $K$ to decrease from $\widetilde I'=f_1(\widetilde{\bf x}')$ to 2 is at most $(\widetilde I'-2)n(n-1)$. Adding them together, the total expected time for $K$ reaching 2 is at most $n(n-1)(n-2)$.

Now, we are ready to finish the proof of the theorem. Suppose ${\bf x}^{A}$ is the individual corresponding to $K=2$, then $f_1({\bf x}^A)=2$ implies that the vertex set corresponding to ${\bf x}^A$ is a CDS. Combining this with $f_1({\bf x}^A)+|{\bf x}^A|\leq m\ln\Delta+2m+2$ implies that $|{\bf x}^A|\leq (2+\ln\Delta)m$, that is, ${\bf x}^A$ approximates the optimal solution within a  factor at most $2+\ln\Delta$. Furthermore, Claim 8 shows that the expected time to obtain such ${\bf x}^A$ is at most $n(n-1)(n-2)$.
\end{proof}

\begin{remark}
{\rm Although our algorithm is essentially SEMO, using GSEMO will also work. Notice that for the second mutation method, the probability that exactly one bit is flipped is ${n\choose 1}\cdot \frac{1}{n} \cdot (1-\frac{1}{n})^{n-1}\geq \frac{1}{e}$. Furthermore, by the way how we define ``good events'' to change indicators from 0 to 1, generating an individual having more bits different from its parent does not affect the validity of the analysis. So, in expected $en(n-1)(n-2)$ iterations, GSEMO can find a $(\ln\Delta+2)$-approximation for MinCDS.}
\end{remark}

\section{Experiments}\label{sec7}

In this section, we experiment on the efficiency and the effectiveness of Algorithm \ref{alg2}, comparing it with the centralized greedy algorithm in \cite{Ruan}. We first carry out experiments on scale-free networks generated by Barab\'{a}si-Albert (BA) model \cite{Barab}: starting from a ring network with four vertices, at each step we add a new vertex with two edges incident to the existing vertices, following the growth and preferential attachment rule. Then we consider the Erd\H{o}s-R\'{e}nyi (ER) network \cite{Peter}: a random graph with $n$ vertices such that any pair of vertices are adjacent with probability $p$. Since the MinCDS problem is studied on connected graphs, to generate a connected graph, we set $p=\frac{\log n}{n}$ \cite{Paul}. The size of the graph is expressed by the number of vertices.

All experiments are coded in Matlab and run on the identical configuration: Intel(R) Core(TM) i5-4210 CPU and 4GB of RAM.

In terms of time, the greedy algorithm is much faster. This is natural, just like the comparatively slow evolution of human being versus fast customized mechanized production. The main purpose of this section is to test the effectiveness of the evolutionary algorithm, which can be achieved within the estimated time guaranteed by the theoretical analysis, and further explore empirical potential for an improvement of time.

We first study the effect of our algorithm on BA networks. The results are shown in Fig. \ref{fig2}. By setting $T$ to be $T_1=n(n-1)(n-2)$, where $n$ is the number of vertices of the input graph, it can be seen from Fig. \ref{fig4} that the size of the solutions output by our algorithm are no worse than those of the greedy algorithm. Interestingly, in many cases, the output of our algorithm is better than that of the greedy algorithm. The reason might lie in the observation that greedy is a deterministic method, which can no longer improve the solution after the algorithm is completed; but for evolutionary algorithm, if more time is allowed, there is still some space for improvement.

Then, we set $T$ to be $T_2=n^2$. As shown by Fig. \ref{fig5}, the accuracy is not decreased compared with $T_1=n(n-1)(n-2)$, which implies that the empirical performance of EA might be much better than what is guaranteed by the theory. This is because our analysis only tracks one evolutionary path having a clearly-cut pattern, which may have underestimated the power of evolution through other paths. We also tried to set $T$ to be $n \log n$, and it turns out that in this case, the output of the algorithm can not be always guaranteed to be a feasible solution.

\begin{figure}[!h]
\centering
\subfigure [EA versus Greedy]{
  \includegraphics[width=6cm]{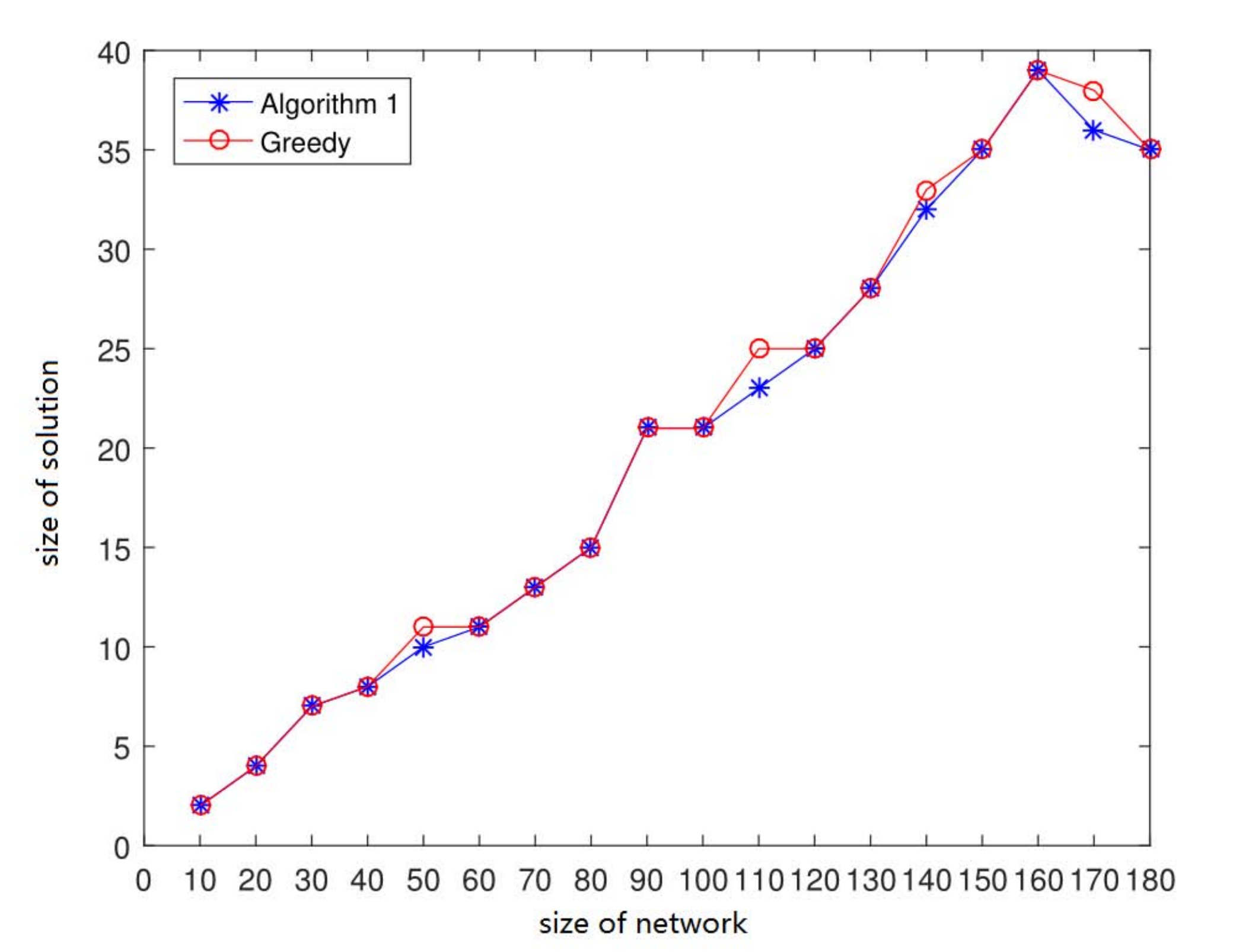} \label{fig4}}
\subfigure [$T_1=n(n-1)(n-2)$ versus $T_2=n^2$]{
  \includegraphics[width=6cm]{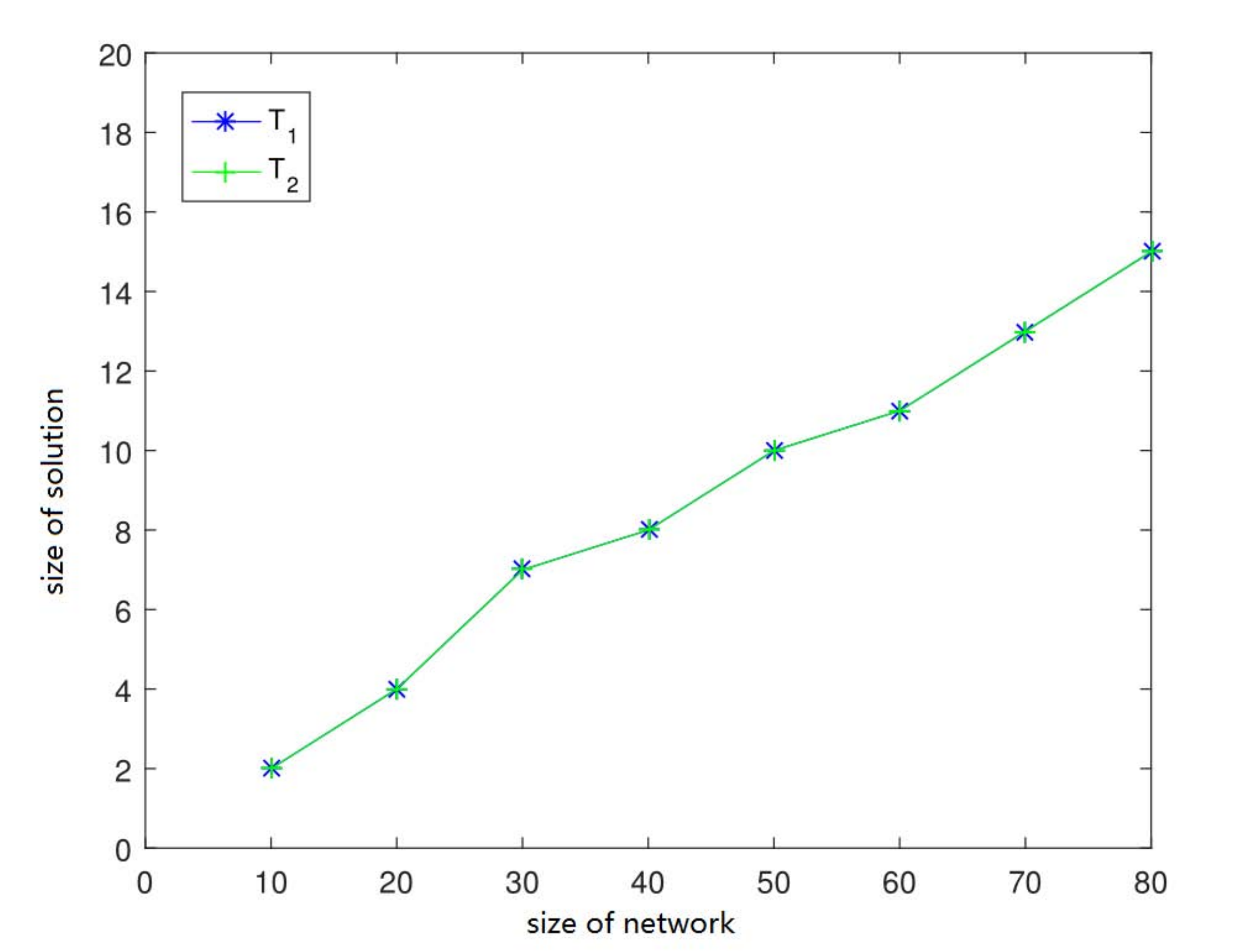} \label{fig5}}
\caption{Experiments on BA networks.}
\label{fig2}
\end{figure}

Then, we consider the performance of Algorithm \ref{alg2} on ER networks. The results are shown in Fig. \ref{fig3}. Similar phenomena are observed for the comparison of EA and greedy (see Fig. \ref{fig6}.  However, divergence appears between $T_1=n(n-1)(n-2)$ and $T_2=n^2$, as shown by Fig. \ref{fig7}. The reason for this phenomenon might be the following: the power-law property possessed by BA networks already makes the solutions good enough. For example, it was shown in \cite{Dinh} that the minimum positive dominating set problem cannot be approximated within factor $(1-\varepsilon)\ln(\max\{\Delta,\sqrt{n}\})$ on a general graph, but can be approximated within a constant factor on a power-law graph. We also tried $T_3=n^2\log n$, the computed sizes coincide with those of $T_1$. So, although the performance of our algorithm on ER is not as good as on BA, there is also potential to improve the time empirically.

\begin{figure}[!h]
 \centering
 \subfigure [EA versus Greedy]{
  \includegraphics[width=6cm]{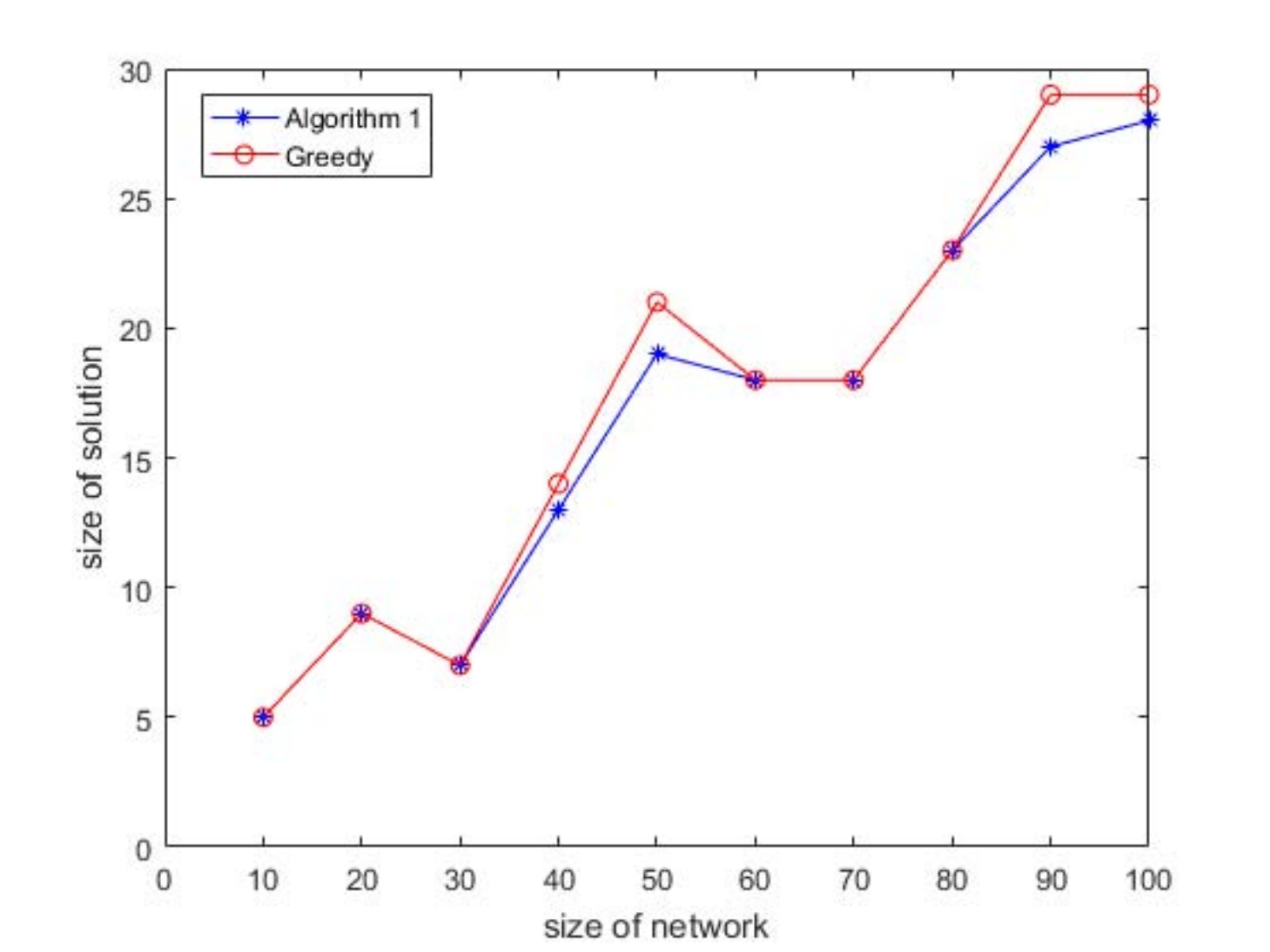} \label{fig6}}
\subfigure [$T_1=n(n-1)(n-2)$ versus $T_2=n^2$]{
  \includegraphics[width=6cm]{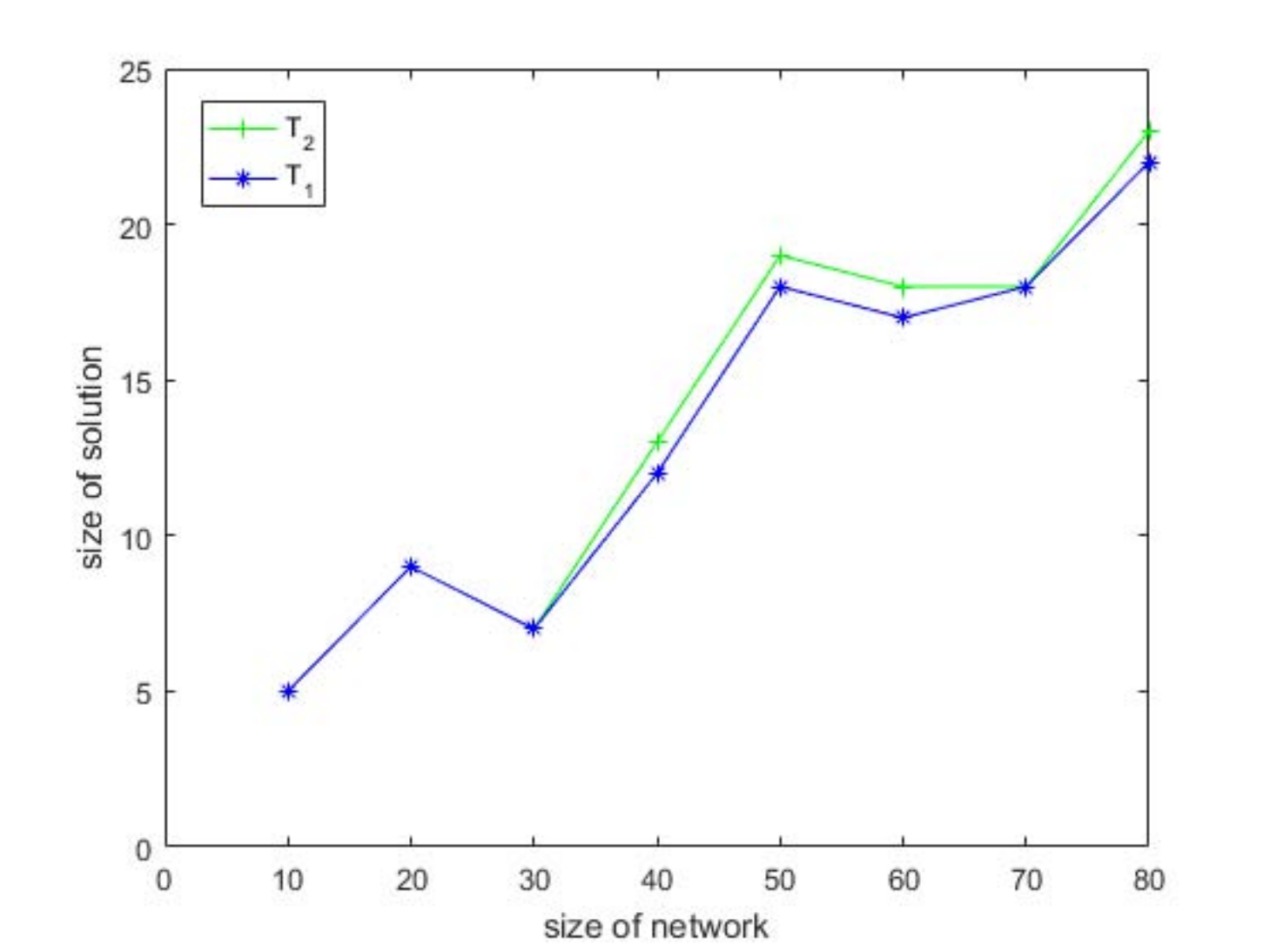} \label{fig7}}
 \caption{Experiments on ER networks.}
 \label{fig3}
 \end{figure}

\section{Conclusion and Discussion}\label{sec6}

This paper shows that a multi-objective evolutionary algorithm can achieve in expected $O(n^3)$ time an approximation ratio $\ln\Delta+2$ for the MinCDS problem, which is the best known ratio that was obtained by a centralized approximation algorithm. The evolutionary algorithm is very simple and easy to implement. However, the theoretical analysis is much more complicated. The most challenging task is how to ``track'' an evolutionary path which is no worse than a greedy path. Our experiments show that a desired approximate solution can indeed be found in the expected time estimated by the theoretical analysis. What is more, EA has the potential to improve the solution if more time resources are available.

Unfortunately, our method could not readily deal with the minimum {\em weight} connected dominating set problem (MinWCDS). In fact, the currently best known approximation ratio for MinWCDS is $1.35\ln n$, obtained by Guha and Khuller in \cite{Guha1999}. This algorithm employs a generalization of an ingenious idea called {\em spider decomposition} which was first invented to study the minimum node-weighted Steiner tree problem \cite{Klein}. However, how to simulate a spider decomposition by an evolutionary algorithm is unclear.

It is interesting to find out more combinatorial optimization problems which can be effectively approximated by an evolutionary algorithm, and try to answer the following question: for which combinatorial optimization problems can evolutionary algorithms achieve a theoretically guaranteed performance, and why?

\section*{Acknowledgements}
This research is supported by NSFC (U20A2068), ZJNSFC (LD19A010001), and NSF (1907472).

	

\end{document}